\documentclass{article}
\usepackage{array,amsmath}
\usepackage{amssymb}
\usepackage{graphicx,subfigure}
\usepackage{graphicx}
\makeatletter
\usepackage{fullpage}
\usepackage{amsthm}
\usepackage{authblk}
\usepackage{algorithm}
\usepackage{verbatim}
\usepackage{algorithmic}

\newtheorem{theorem}{Theorem}[section]
\usepackage{fullpage}
\newtheorem{lemma}[theorem]{Lemma}
\newtheorem{corollary}{Corollary}[theorem]
\makeatother

\newcommand{\old}[1]{{}}
\title{Approximation Algorithms for Anchored Multiwatchman Routes\thanks{A preliminary version of this paper was entitled ``Multirobot Watchman Routes in a Simple Polygon'' and appears in the Proceedings of the 36th Canadian Conference on Computational Geometry, St. Catharines, Ontario, Canada, pp. 49-55, 2024.}}
\author{Joseph S. B. Mitchell\thanks{Department of Applied Mathematics and Statistics, Stony Brook University, \texttt{joseph.mitchell@stonybrook.edu}}\hspace{0.2\textwidth}Linh Nguyen\thanks{Department of Applied Mathematics and Statistics, Stony Brook University, \texttt{linh.nguyen.1@stonybrook.edu}}}
\date{}

\begin{document}

\maketitle

\begin{abstract}
We study some variants of the $k$-\textsc{Watchman Routes} problem, the cooperative version of the classic \textsc{Watchman Routes} problem in a simple polygon. The watchmen may be required to see the whole polygon, or some pre-determined quota of area within the polygon, and we want to minimize the maximum length traveled by any watchman. While the single watchman version of the problem has received much attention is rather well understood, it is not the case for multiple watchmen version.

We provide the first tight approximability results for the anchored $k$-\textsc{Watchman Routes} problem in a simple polygon, assuming $k$ is fixed, by a fully-polynomial time approximation scheme. The basis for the FPTAS is provided by an exact dynamic programming algorithm. If $k$ is a variable, we give constant-factor approximations. 
\end{abstract}

\section{Introduction}
\subsection{Prior and Related Work}
The \textsc{Art Gallery} problem was introduced in 1973 by Victor Klee: given an art gallery with $n$ walls (a polygon $P$), determine the minimum number of stationary guards needed stationed within $P$ so that every point inside the gallery is visible to at least one guard. Since then, the \textsc{Art Gallery} problem, along with its many variants, has been the focus of extensive research in computational geometry and algorithms.

When guards are mobile, a single guard can monitor any connected domain. Consequently, the objective shifts to finding routes for one or more guards that optimize certain aspects of their movement, such as path lengths or the number of turns. The problem of minimizing the distance a guard must travel to cover the entire polygon is known as the \textsc{Watchman Route} problem. Chin and Ntafos introduced the \textsc{Watchman Route} problem, proved that it is NP-hard in polygons with holes (later revisited in~\cite{dumitrescu2012watchman}), and provided a linear time algorithm for simple orthogonal polygons. For general simple polygons, there are exact polynomial-time algorithms, for example, \cite{dror2003touring} gives an $O(n^3\log n)$ algorithm for the anchored version (where a starting point $s$ on the route is specified) and an $O(n^4\log n)$ algorithm for the floating version (where no starting point is given). Recently, a slightly better algorithm for the anchored version with $O(n^3)$ running time has also been proposed~\cite{tan2018touring}. In a polygon with holes, Mitchell showed that the \textsc{Watchman Route} problem cannot be approximated better than an $O(\log n)$ factor, and gave an $O(\log^2n)$-approximation.

In some scenarios, complete coverage may be impractical or unnecessary, so we are also interested in finding the shortest route that covers at least an area of $A \geq 0$ within $P$. This variant is known as the \textsc{Quota Watchman Route} problem, introduced in~\cite{huynh_et_al:LIPIcs.SWAT.2024.27}. In contrast to the tractable \textsc{Watchman Route} problem, the \textsc{Quota Watchman Route} problem is weakly NP-hard, but a fully polynomial-time approximation scheme (FPTAS) is known. Any results for the \textsc{Quota Watchman Route} problem can be applied to the \textsc{Watchman Route} problem by simply setting $A$ to be the area of $P$.

We consider the generalization of both aforementioned problems to multiple agents, the $k$-\textsc{Watchman Routes} problem and the \textsc{Quota $k$-Watchman Routes} problem, with the objective of minimizing the length of the longest path traveled by any one watchman. Even in a simple polygon, when no starting points are specified (so, we are to determine the best starting locations), both problems are NP-hard to approximate within any multiplicative factor~\cite{packer2008computing}. 

Our focus is thus on the anchored version, where a team of robots or searchers enter a domain $P$ through a door on its boundary to search for a stationary target that may be randomly distributed within the domain. The goal is to plan an optimal collective effort to ensure at least a certain probability of detection (1 in the $k$-\textsc{Watchman Routes} problem and some $p \in [0, 1]$ in the \textsc{Quota $k$-Watchman Routes} problem).

Little is known about the $k$-\textsc{Watchman Routes} problem in its most general form. Very recently, Nilsson and Packer~\cite{nilsson2024approximation} proposed some constant-factor approximation algorithms for when $k=2$, their methods however do not extend to larger values of $k$. Additionally, some studies have focused on developing and analyzing experimental heuristics~\cite{faigl2010approximate, packer2008computing}, but these approaches lack provable performance guarantees. It was clear that many algorithmic results for the \textsc{Watchman Routes} problem do not generalize to the $k$-\textsc{Watchman Routes} problem.

\subsection{Main Results}
This paper is an extended version of our earlier work that appeared in preliminary form in~\cite{mitchell2024multirobot}, and includes some new findings. The main results can be summarized as follows:
 \begin{itemize}
     \item Assuming that $k$ is fixed and relatively small, we give a pseudopolynomial-time exact algorithm to solve the anchored $k$-\textsc{Watchman Routes} problem in a simple orthogonal (integral coordinate) polygon $P$ under $L_1$ distance. This is a reasonable assumption as in most practical situations, it is not feasible to employ an arbitrarily large number of watchmen or robots.
     \item For any fixed $k$, we acquire an FPTAS for the anchored $k$-\textsc{Watchman Routes} problem in an orthogonal simple polygon (with orthogonal watchmen movements) and a general simple polygon. Since the anchored $k$-\textsc{Watchman Routes} problem is known to be weakly NP-hard~\cite{mitchell1991watchman}, we obtain tight approximability.
     \item For variable $k$, we give constant-factor approximation algorithms, for the anchored $k$-\textsc{Watchman Routes} problem as well as the anchored \textsc{Quota $k$-Watchman Routes} problem. The algorithms given are simple and have a low approximation factor.
 \end{itemize}
 While we restrict ourselves to the anchored version, we achieve better approximation factors for any $k$ than those proposed by Nilsson and Packer for the case where $k=2$~\cite{nilsson2024approximation}.

\section{Preliminaries}
\label{sec:prelims}
Our domain is a simple polygon $P$, whose boundary, $\partial P$ is a simple, non-self-intersecting polygonal chain, consisting of $n$ vertices. A vertex is \textit{reflex} (resp. \textit{convex}) if its internal angle is greater than (resp. less than) 180 degrees. An \textit{orthogonal simple polygon} is a simple polygon where all internal angles are 90 or 270 degrees.

We say that $x\in P$ and $y\in P$ \textit{see} each other if the segment $xy$ is contained inside $P$. The \textit{visibility region} of $x$, denoted by $V(x)$, is defined as the set of all points that $x$ sees, i.e. $V(x) = \{y\in P: xy \subseteq P\}$. The visibility region of a segment or route (or any subset) $X\subseteq P$ is the union of the visibility regions of all points of $X$; or equivalently the set of all points that see some point of $X$. Efficient algorithms for computing visibility regions of points and segments are known, see~\cite{guibas1986linear, toth2017handbook}.

The \textit{geodesic shortest path} from $x\in P$ to $y\in P$, which we denote by $\pi(x,y)$ is a shortest path between $x$ and $y$ that stays within $P$. If $P$ is a simple polygon, then $\pi(x,y)$ is necessarily unique and polygonal. The \textit{geodesic distance} between $x$ and $y$ is the length of $\pi(x,y)$ measured in the $L_2$ metric, also called the Euclidean metric. The geodesic shortest path from $x$ to a set $Y\subseteq P$ is $\pi(x,y)$ such that the length of $\pi(x,y)$ is the minimum among all $y\in Y$. Throughout this paper, if we use the term ``geodesic shortest path'', we mean the Euclidean geodesic shortest path. The \textit{geodesic convex hull} or \textit{relative convex hull} of $X\subseteq P$ is the minimal subset of
$P$ that contains $X$ and is closed under taking shortest paths; one important property of the relative convex hull is that it is the minimum-perimeter superset of $X$ within $P$. We also consider another metric, namely the $L_1$ metric, also known as the Manhattan metric. A polygonal path with all edges parallel to an axis is called a rectilinear path. A rectilinear path of minimum length between two points that stays within $P$ is called a \textit{geodesic $L_1$ shortest path}. Unlike the Euclidean metric, there are, in general, multiple geodesic $L_1$ shortest paths between two points, all of the same length. We denote by $\pi^{\perp}(x,y)$ a geodesic $L_1$ shortest paths from $x$ to $y$ . Refer to the survey~\cite{mitchell2000geometric} for more details about geodesic shortest paths and relative convex hulls.

The \textsc{Watchman Route} problem asks for a minimum-length tour that sees all of $P$. A \textit{visibility cut} $c_i$ with respect to the starting point $s$ is a chord obtained from extending the edge $e$ incident on a reflex vertex, $v_i$, where $e$ is such that its extension creates a convex vertex at $v_i$ in the subpolygon containing $s$. The other subpolygon (not containing $s$) is the \textit{pocket} induced by $c_i$. Furthermore, $c_i$ is an \textit{essential cut} if its pocket does not fully contain any other pocket. The shortest visibility covering tour must visit all essential cuts in order around $\partial P$ \cite{chin1991shortest, dror2003touring}, thus by treating the \textsc{Watchman Route} problem as an instance of the ``touring polygons'' problem, we can get a polynomial-time algorithm~\cite{dror2003touring}.

The \textsc{Quota Watchman Route} problem seeks a minimum-length tour that sees at least a given amount of area within $P$. The QWRP is (weakly) NP-hard, and the touring polygons approach does not apply since we do not have the optimality structure implied by having to see all of $P$, which yields a sequence of essential cuts (the watchman may only need to look ``partially'' around some corners to achieve the quota) for the tour to visit. By exploiting the following geometric property: an optimal (complete or partial) visibility covering tour must be relatively convex (the tour as well as its interior are closed under taking geodesic shortest paths), we can separate the visibility region of different intervals of the tour. This allows for a simple dynamic programming algorithm on an appropriate discretization of the domain, leading to an FPTAS~\cite{huynh_et_al:LIPIcs.SWAT.2024.27}.

We consider the generalization of both of the above problems to multiple watchmen. Denote by $|.|$ the measure of geometric objects (such as length or area). Given a simple polygon $P$ and a starting point $s$. The $k$-\textsc{Watchman Routes} problem is that of computing a collection of $k$ tours $\{\gamma_i\}_{i=1,\ldots, k}$ such that $\bigcup\limits_{i=1,\ldots,k}V(\gamma_i) = P$ and $\max\limits_{i=1,\ldots,k}|\gamma_i|$ is minimized. Note that each $\gamma_i$ is necessarily a polygonal tour (Corollary~\ref{cor:shortest}).

In any multi-agent collaboration problem, the two commonly studied objective functions are min-max and min-sum cost, however in the context of the $k$-\textsc{Watchman Routes} problem with a common starting point for all watchmen, the min-sum version is trivially solvable: simply take the relative convex hull of $\{\gamma_i\}_{i=1,\ldots, k}$, its perimeter is no greater than $\sum_{i=1}^k|\gamma_i|$ since all routes are connected (through $s$), thus we can let one watchman travel the perimeter of said relative convex hull to see all of $P$ while all others stand still. In other words, the optimal solution for the min-sum objective is a shortest route for a single watchman and $k-1$ degenerate routes of length 0. If the perimeter of the relative convex hull of $\{\gamma_i\}_{i=1,\ldots, k}$ does not pass through $s$, instead of said relative convex hull, we can simply take the (``nearly relatively convex'') minimum-perimeter set enclosing $\{\gamma_i\}_{i=1,\ldots, k}$ and containing $s$ as a vertex, which is computable if given $s$ and the set $\{\gamma_i\}_{i=1,\ldots, k}$, and proceed with a similar argument. The min-max version remains NP-hard, even with a common depot.

The next problem we study, the \textsc{Quota $k$-Watchmen Routes} problem has the same min-max objective function, but with a generalized constraint of $\left|\bigcup\limits_{i = 1,\ldots, k}V(\gamma_i)\right| \ge A$ for some given $0 \le A \le |P|$. The fraction of area seen, $\frac{A}{|P|}$, can be interpreted as the probability that the watchmen detect a target at an unknown location uniformly distributed in $P$.

\section{Multiwatchman Routes With a Constant Number of Watchmen}
First, we characterize the structure of an optimal collection of $k$-watchman routes $\{\gamma_i\}_{i=1,\ldots, k}$.

\begin{lemma}
    \label{lem:visit_all_essential_cuts}
    $\bigcup\limits_{i=1,\ldots,k}V(\gamma_i) = P$ if and only if $\{\gamma_i\}_{i=1,\ldots, k}$ collectively visit all essential cuts of $P$.
\end{lemma}

\begin{proof}
    See Figure~\ref{fig:visit_all_essential_cuts}. Since all routes pass through $s$, their union is connected. The proof follows simply from the well known fact: A connected set sees all of $P$ if and only if it visits all essential cuts \cite{carlsson1999finding, chin1991shortest, dror2003touring}.
\end{proof}

\begin{figure}[!h]
    \centering
    \includegraphics[width=0.75\linewidth]{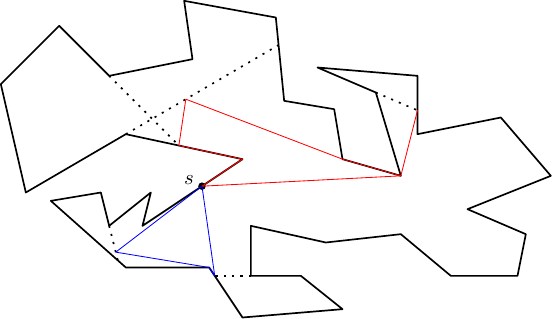}
    \caption{A simple polygon and its essential cuts (dashed). The 2-watchman routes (red and blue) visit all essential cuts between them and see the whole polygon.}
    \label{fig:visit_all_essential_cuts}
\end{figure}

Note that when we have multiple starting points, or no starting points specified at all, Lemma~\ref{lem:visit_all_essential_cuts} does not necessarily hold.

For $1 \le i \le k$, denote by $C_i$ the set of essential cuts visited by $\gamma_i$. It is worth pointing out that, due to Lemma~\ref{lem:visit_all_essential_cuts}, $\gamma_i$ behaves like a solution to the touring polygons problem on a subset of essential cuts.
\begin{corollary}
    \label{cor:shortest}
    There exists an optimal solution $\{\gamma_i\}_{i=1,\ldots, k}$ such that for any $i$, $\gamma_i$ is the shortest route to visit all cuts in $C_i$ and $s$ in the order in which they appear around $\partial P$ (more precisely, the order of the reflex edges whose extensions constitute the essential cuts).
\end{corollary}
One immediate observation is that each $\gamma_i$ is necessarily a closed polygonal curve, with at most $2n$ vertices~\cite{dror2003touring, mitchell2013approximating}. It is easy to show that if $i\ne j$, then $C_i \cap C_j = \varnothing$. Suppose to the contrary that, there exists $c\in C_i \cap C_j$, then we can replace $\gamma_i$ by the shortest tour that visits all of $C_i\setminus c$, which is no longer than the shortest tour that visits all of $C_i$.

\subsection{Orthogonal simple polygon, rectilinear movement}
We first consider the special case where $P$ is a simple, orthogonal polygon; without loss of generality, assume all edges of $P$ are either horizontal or vertical and all vertices of $P$ have integer coordinates. Further assume that the watchmen all start from a common starting point $s\in \partial P$ and can only move along axis-aligned segments. It is known that even for $k = 2$, the general $k$-\textsc{Watchman Routes} in a simple polygon is (weakly) NP-hard via a simple reduction from the \textsc{Partition} problem~\cite{mitchell1991watchman} (given a set of positive integers, partition it into two subsets of equal sum). The reduction can be easily modified to show that the anchored orthogonal version we consider here is also NP-hard.

\subsubsection{Dynamic programming exact algorithm} We decompose $P$ into rectangular cells using maximal (within $P$) extensions of all edges, as well as a horizontal and vertical line through $s$. This is known as the Hanan grid~\cite{hanan1966steiner} (Figure~\ref{fig:HanagridinP}). We refer to the set of vertices of the rectangular cells, including vertices of $P$ as grid points.

\begin{figure}[!h]
		\centering
		\includegraphics[scale=0.9]{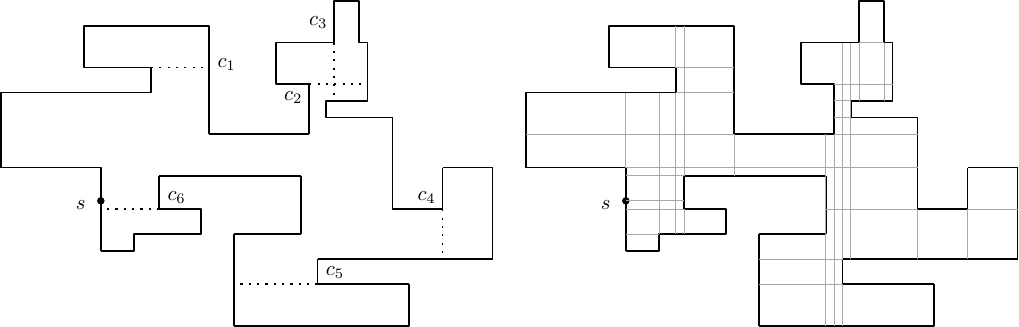}
		\caption{\textit{Left}: The essential cuts of $P$. \textit{Right}: The Hanan grid formed by extensions of all edges in $P$.}
		\label{fig:HanagridinP}
	\end{figure}

We show that, given the assumptions of an orthogonal $P$ and rectilinear movement, the Hanan grid (or in particular, the portion of the Hanan grid on the essential cuts) suffices to discretize the problem.
 \begin{lemma}
 \label{lem:discrete_viewpoints}
     There exists an optimal solution $\{\gamma_i\}_{i=1,\ldots, k}$ that lies on the Hanan grid.
 \end{lemma}

 \begin{proof}
     Given an optimal solution $\{\gamma_i\}_{i=1,\ldots, k}$, for $1\le i\le k$ let $C_i = \{c_{i1}, \ldots, c_{ij}\}$ (in order around $\partial P$). Also, let $p_{i1}, \ldots, p_{ij}$ be the point where $\gamma_i$ first intersects with $c_{i1}, \ldots, c_{ij}$.
     
     First, note that for every $i$, we may replace $\gamma_i$ with a concatenation of geodesic $L_1$ shortest paths, namely $\gamma_i:=\pi^\perp(s, p_{i1})\cup \pi^\perp(p_{i1}, p_{i2})\cup \ldots\cup \pi^\perp(p_{ij}, s)$ without increasing $\max\limits_{i = 1,\ldots,k}|\gamma_i|$ while maintaining complete visibility coverage of $P$.
     
     Now, we argue that $\pi^\perp(s, p_{i1})$ is a geodesic $L_1$ shortest path from $s$ to $c_{i1}$. Suppose to the contrary, that geodesic $L_1$ shortest paths from $s$ to $c_{i1}$ make contact with $c_{i1}$ at $p'_{i1} \ne p_{i1}$ (all geodesic $L_1$ shortest paths from a point to a segment have the same endpoint). Due to orthogonality $|\pi^\perp(s,p'_{i1})| + |p'_{i1}p_{i1}| = |\pi^\perp(s,p_{i1})|$, which means $|\pi^\perp(s,p_{i1})| + |\pi^\perp(p_{i1}, p_{i2})| = |\pi^\perp(s,p'_{i1})| + |p'_{i1}p_{i1}| + |\pi^\perp(p_{i1}, p_{i2})| \ge |\pi^\perp(s,p'_{i1})| + |\pi^\perp(p'_{i1}, p_{i2})|$. This implies $\gamma_i$ should take a geodesic $L_1$ shortest path from $s$ to $c_{i1}$, and it suffices to find such a path in the Hanan grid. By a straightforward inductive argument, we can show the same for any portion of $\gamma_i$ between any two essential cuts.
 \end{proof}

 Lemma~\ref{lem:discrete_viewpoints} allows us to treat the orthogonal $k$-watchman problem as that of finding a set of grid points on the essential cuts (one per cut) that each route ${\gamma_i}$ visits so as to minimize the maximum length of any route. Note that we need not consider any grid point that lies beyond an essential cut.

For the purpose of describing our algorithm, it may be more convenient to think of each route $\gamma_i$ in terms of two segments: the path from $s$ to $p_{ij}$, which we denote by $\Gamma_i$, and the return trip $\pi^\perp(p_{ij},s)$. Formally, $\Gamma_i:=\pi^\perp(s, p_{i1})\cup \pi^\perp(p_{i1}, p_{i2})\cup \ldots\cup \pi^\perp(p_{i(j-1)}, p_{ij})$. 

Our algorithm is based on dynamic programming. Let $\{c_1, c_2, \ldots, c_m\}$ be the set of essential cuts in order around $\partial P$ ($s$ lies between $c_1$ and $c_m$). Each subproblem $(c_j, p_1, l_1, \ldots, p_k, l_k)$ is defined by a $(2k+1)$-tuple which consists of
\begin{itemize}
    \item an essential cut $c_j$,
    \item $k$ Hanan grid points $p_1, \ldots, p_k$, each on one different essential cut among $s, c_1, \ldots, c_j$. One of $p_1, \ldots, p_k$ is on $c_j$.
    \item $k$ integers $l_1, \ldots, l_k$.
\end{itemize}
We abuse notations slightly, for the sake of brevity, by using the tuple to also denote the value of subproblem in our dynamic programming algorithm. Subproblem $(c_j, p_1, l_1, \ldots, p_k, l_k)=$ TRUE if and only if there exists a collection of $k$ paths $\Gamma_1, \ldots, \Gamma_k$ collectively visiting all essential cuts from $c_1$ up to $c_j$ such that
 \begin{itemize}
     \item $\Gamma_i$ starts at $s$, ends at $p_i$, and consists of geodesic $L_1$ shortest paths between grid points on essential cuts,
     \item $|\Gamma_i| = l_i$.
 \end{itemize}
We call $\{\Gamma_i\}_{i=1,\ldots, k}$ the paths associated with subproblem $(c_j, p_1, l_1, \ldots, p_k, l_k)$. See Figure~\ref{fig:subproblem} for an illustration.
 \begin{figure}[h]
    \centering
    \includegraphics[width=0.65\textwidth]{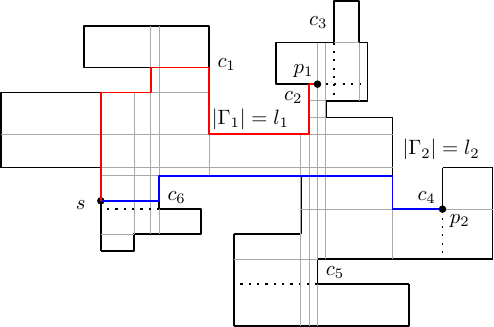}
    \caption{An example subproblem $(c_4, p_1, l_1, p_2, l_2)=$ TRUE, $\Gamma_1$ (resp. $\Gamma_2$) is drawn in red (resp. blue).}
    \label{fig:subproblem}
\end{figure}

For each Hanan grid point $p\in c_j$ and $i = 1, \ldots, k$, the value of the subproblem $(c_j, p_1, l_1, \ldots, p_i:=p, l_i, \ldots, p_k, l_k)$ is computed according to the following disjunction of valid subproblems
\begin{align}
\label{eqn:recursion_DP}
    (c_j, p_1, l_1, \ldots, p_i:=p, l_i, \ldots, p_k, l_k) = \bigvee\limits_{p'}(c_{j-1},p_1, l_1, \ldots, p_i:=p', l_i - |\pi^\perp(p, p')|, \ldots, p_k, l_k)
\end{align}
where $p'$ is taken from the set of all Hanan grid points on the cuts $c_1, \ldots, c_{j-1}$ such that geodesic $L_1$ shortest paths from $p'$ to $c_j$ make contact with $c_j$ at $p$ (Lemma~\ref{lem:discrete_viewpoints}). See Figure~\ref{fig:dp_recursion} for a concrete example of the above recursion.

\begin{figure}[h]
    \centering
    \includegraphics[scale=0.9]{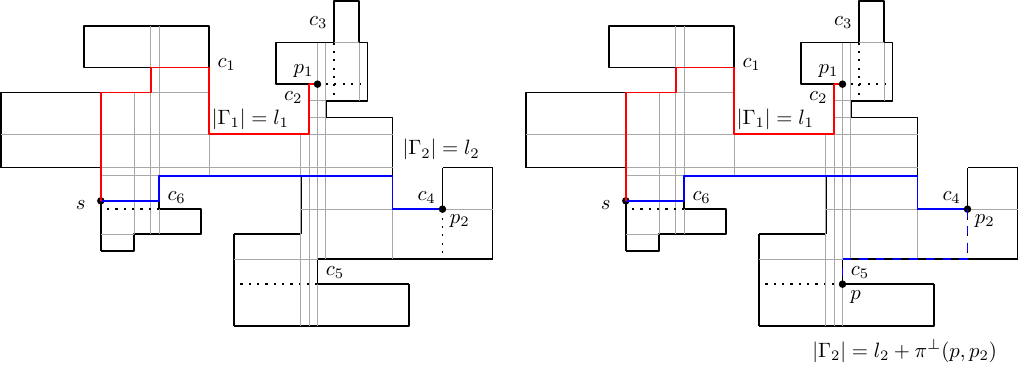}
    \caption{If $(c_4, p_1, l_1, p_2, l_2)=$ TRUE (right), then $(c_5, p_1, l_1, p, l_2 + \pi^\perp(p,p_2))=$ TRUE (left).}
    \label{fig:dp_recursion}
\end{figure}

The base case is simply $(s, s,0,\ldots,s,0) = \text{TRUE}$. From there, we tabulate all possible subproblems, take the subproblem $(c_m, p_1, l_1, \ldots, p_k, l_k)$ whose value is TRUE with the minimum $\max\limits_{i=1,\ldots, k}\{l_i + |\pi^\perp(p_i, s)|\}$ and return the associated paths, each concatenated with the corresponding return trip to the starting point, specifically $\{\gamma_i := \Gamma_i \cup \pi^\perp(p_i, s)\}_{i=1,\ldots, k}$. This gives us an exact optimal solution to the $k$-\textsc{Watchman Routes} problem in an orthogonal polygon, where the watchmen are restricted to rectilinear movement.

Our proof of correctness relies on two arguments:
 \begin{itemize}
     \item If the paths associated with subproblem $(c_{j-1},p_1, l_1, \ldots, p_i:=p', l_i - |\pi^\perp(p, p')|, \ldots, p_k, l_k)$ visit $s, c_1, \ldots, c_{j-1}$, then those
     associated with $(c_j, p_1, l_1, \ldots, p_i:=p, l_i, \ldots, p_k, l_k)$ visit all essential cuts from $s$ to $c_j$ since $p\in c_j$. By induction, the tours $\{\gamma_i\}_{i=1,\ldots, k}$ returned by the algorithm hence visit all essential cuts.
     \item $\gamma_i$ consists of geodesic $L_1$ shortest paths between contact points with essential cuts (proof of Lemma~\ref{lem:discrete_viewpoints}). If we identify two consecutive contact points on $\gamma_i$, say $p'$ and $p$ in that order, then the length of the portion of $\gamma_i$ from $s$ to $p$ is $l_i$ if and only if the length of the portion of $\gamma_i$ from $s$ to $p'$ is $l_i - |\pi^\perp(p, p')|$.
 \end{itemize}

 \subsubsection{Analysis of running time} First we bound the number of subproblems in the dynamic program. There are $O(n)$ essential cuts, $O(n)$ Hanan grid points on each cut. The length of each tour $\gamma_i$ can be trivially upper bounded by the perimeter of $P$, $|\partial P|$. Hence, in total there are $O\left(n\cdot n^{2k}\cdot |\partial P|^k\right) = O\left(n^{2k +1}|\partial P|^k\right)$ subproblems.
 
 We pre-compute and store geodesic $L_1$ shortest paths between Hanan grid points, as well as between Hanan grid points and essential cuts by solving the \textsc{All Pairs Shortest Path} problem in the embedded graph of the Hanan grid; the running time of this step is clearly dominated by the dynamic program. Then, we can solve each subproblem by iterating through $O(n^2)$ previously solved subproblems, using recursion~\eqref{eqn:recursion_DP}. The overall running time of the dynamic program is thus $O(n^{2k +3}|\partial P|^k)$, which is pseudopolynomial for any fixed $k$. This is in congruence with the weak NP-hardness from \textsc{Partition}, for which there exists an algorithm polynomial in the number of input integers and the largest input integer.

 We can also achieve a tighter, however straightforward, time bound, which is useful for deriving the FPTAS. Let $L$ be the length of a shortest single orthogonal watchman route of $P$, which is computable in $O(n)$ time if $P$ is simple and orthogonal (simply use the $L_1$ shortest path in the algorithm given in~\cite{CN} instead of the usual $L_2$ shortest path). Clearly $L \le |\partial P|$ and $\min\max\limits_{i=1,\ldots,k}|\gamma_i| \le L$ (one shortest single watchman route and $k-1$ routes of length 0 is a feasible solution to the $k$-\textsc{Watchman Routes} problem). Thus, we get a time bound of $O(n^{2k+3}L^k)$. 

 \subsubsection{Fully polynomial-time approximation scheme} The running time of the dynamic programming algorithm depends on the maximum possible length of the routes. To achieve fully polynomial running time for fixed $k$, we bound the number of subproblems by a ``bucketing'' technique.

 Recall in Section~\ref{sec:prelims}, we showed that the min-sum version with a common depot is readily solvable in polynomial time by arguing $L \le \sum_{i=1}^k|\gamma_i|$. Thus
 \begin{align}
\label{inq:lower_bound}
    \frac{L}{k}\le \max\limits_{i=1,\ldots,k}|\gamma_i|\le L.
\end{align}
Given any $\varepsilon > 0$, we divide $L$ into $\lceil\frac{nk}{\varepsilon}\rceil$ uniform intervals, each must be no longer than $\frac{\varepsilon L}{nk}$. The length of any geodesic $L_1$ shortest path we take into consideration for recursion \eqref{eqn:recursion_DP} must fall into one of the intervals (otherwise, the upper bound $\max\limits_{i=1,\ldots,k}|\gamma_i|\le L$ is violated). We round the length of the aforementioned path down to the nearest interval endpoint. Then, apply the dynamic programming algorithm to the new instance with a slight modification: let the intervals' endpoints define the subproblems instead of the $k$ integers $l_1, l_2, \ldots, l_k$. Denote by $\{\gamma_i'\}_{i=1,\ldots, k}$ the solution returned. For clarity, we denote by $d(.)$ distance/length in the ``rounded down'' instance. We have the following chain of inequalities 
\begin{align}
\label{inq:definition}
    \max\limits_{i=1,\ldots,k}|\gamma_i| \ge \max\limits_{i=1,\ldots,k}d(\gamma_i) \ge \max\limits_{i=1,\ldots,k}d(\gamma_i').
\end{align}
The first inequality follows simply from the fact that we round down any distance from the original instance. The second inequality is by definition, since $\{\gamma_i'\}_{i=1,\ldots, k}$ is an optimal solution of the new instance.

Now, any route in $\{\gamma_i'\}_{i=1,\ldots, k}$ consists of at most $n$ geodesic $L_1$ shortest paths between Hanan grid points on essential cuts, the length of each differs by no more than $\frac{\varepsilon L}{nk}$ between the original instance and the ``rounded down'' instance. Thus, for any $i$
\[|\gamma_i'| - d(\gamma_i') \le n\cdot\frac{\varepsilon L}{nk} \]
and therefore
\begin{align}
\label{inq:rounding}\max\limits_{i=1,\ldots,k}d(\gamma_i') + \frac{\varepsilon L}{k} \ge \max\limits_{i=1,\ldots,k}|\gamma_i'|.\end{align}
Combining \eqref{inq:lower_bound}, \eqref{inq:definition}, \eqref{inq:rounding},  we get 
\[(1 + \varepsilon)\max\limits_{i=1,\ldots,k}|\gamma_i| \ge \max\limits_{i=1,\ldots,k}|\gamma_i'|.\]
Thus, we achieve an FPTAS with a running time of $O\left(n^{2k+3}\cdot\left(\frac{nk}{\varepsilon}\right)^k\right) = O\left(\frac{n^{3k+3}}{\varepsilon^k}\right)$.

If the watchmen are not constrained to rectilinear movement, then the dynamic programming algorithm and the FPTAS above give a $\sqrt{2}$-approximation and a $(\sqrt{2} + \varepsilon)$-approximation, respectively, for unconstrained movement ($L_2$ metric).

\begin{theorem}
    For any fixed $k$, the anchored $k$-\textsc{Watchman Routes} problem in a simple orthogonal polygon has an FPTAS for the $L_1$ metric and a polynomial-time $(\sqrt{2} + \varepsilon)$-approximation for the $L_2$ metric. 
\end{theorem}

\subsection{General simple polygon, unconstrained movement}
In our dynamic programming algorithm for the orthogonal anchored $k$-\textsc{Watchman Routes} problem, by restricting to rectilinear movement and axis-aligned essential cuts, we have a discrete set of points at which the routes can come into contact with the essential cuts. We extend that idea to the case in which $P$ is a general simple polygon and the watchmen can travel in any directions with an appropriate  approach to discretization, to achieve an FPTAS. Naively discretizing the whole of $P$ into a fine grid can be far too computationally expensive. We overcome this challenge by showing how to localize the problem. 

\subsubsection{Discretization of $P$} Denote by $GD_s(r)$ the geodesic disk of radius $r$ centered at the given starting point $s$, i.e. the set of all points with geodesic distance (length of the geodesic shortest path) no greater $r$ from $s$. Let $r=r_{\min}$, where $r_{\min}$ is the smallest value of $r$ such that $V(GD_s(r)) = P$, i.e. $GD_s(r)$ sees everything in $P$ or equivalently, $\partial GD_s(r)$ is a visibility covering tour of $P$. To determine $r_{\min}$, simply compute the largest geodesic distance from $s$ to any of the $O(n)$ essential cuts, then $GD_s(r_{\min})$ touches every essential cut and sees all of $P$. Consider an optimal collection of tours $\{\gamma_i\}_{i=1,\ldots, k}$ with the minimum $\max\limits_{i=1,\ldots,k}|\gamma_i|$.

\begin{lemma}
\label{lem:bounding_general_poly_see_everything}
    $r_{\min}\le \max\limits_{i=1,\ldots,k}|\gamma_i| \le 6nr_{\min}$.
\end{lemma}
\begin{proof}
    The geodesic disk $GD_s\left(\frac{\max\limits_{i=1,\ldots,k}|\gamma_i|}{2}\right)$ contains $\{\gamma_i\}_{i=1,\ldots, k}$, and therefore must see all of $P$ (since $P$ is simple and hole-free). By definition $r_{\min}$ is the smallest radius such that $GD_s(r_{\min})$ sees all of $P$, so $r_{\min} \le\frac{\max\limits_{i=1,\ldots,k}|\gamma_i|}{2}$.

    In this proof, let $L$ be the length of the shortest unconstrained (direction-wise) watchman route of $P$ through $s$. Clearly, $L$ is an upper bound on $\max\limits_{i=1,\ldots,k}|\gamma_i|$. We show that $L \le 6nr_{\min}$. Indeed, since $\partial GD_s(r_{\min})$ sees all of $P$, $L \le |\partial GD_s(r_{\min})| + 2r_{\min}$ (the left hand side is the most cost incurred by traveling around $\partial GD_s(r_{\min})$ and the shortest path from $s$ to $\partial GD_s(r_{\min})$ twice in both directions).

    Finally, observe that $\partial GD_s(r_{\min})$ consists of polygonal chains that are portions of $\partial P$ and circular arcs; the circular arcs have total length no greater than $2\pi r_{\min}$. Each segment in the polygonal part of $\partial GD_s(r_{\min})$ has its length bounded by the sum of geodesic distances from its endpoints to $s$ (triangle inequality), which is no more than $2r_{\min}$. There are at most $n$ segments in the polygonal portions of $\partial GD_s(r_{\min})$, therefore their total length is no greater than $2nr$, implying $|\partial GD_s(r_{\min})| + 2r_{\min} \le 2nr_{\min} + 2\pi r_{\min} + 2r_{\min} \le 6nr_{\min}$.
\end{proof}
Starting from $r = r_{\min}$, if we repeatedly multiple $r$ by 2, at some point we must have $r\le\max\limits_{i=1,\ldots,k}|\gamma_i|\le 2r$, suppose we have reached this point. Consider a regular square grid of pixels of side lengths $\delta$ (we specify $\delta$ below) within an axis-aligned closed square box $B$ of size $r$-by-$r$ centered on $s$. The restriction of $B$ within $P$, i.e. $P\cap B$ contains an optimal $\{\gamma_i\}_{i=1,\ldots, k}$, note that this means $B$ necessarily intersects all essential cuts. We triangulate $P$ (using any of the numerous known algorithms, for example, the linear-time algorithm in~\cite{chazelle1991triangulating}), including $s$ as a vertex of the triangulation. The arrangement of the square grid, the triangulation and the essential cuts decomposes $P\cap B$ into convex cells, each of perimeter no greater than $4\delta$ (Figure~\ref{fig:discretization_general_simple_polygon}). The set of vertices of this decomposition is composed of intersections between chords of the triangulation, the essential cuts and the axis-aligned grid lines, as well as vertices of $P$ and the starting point $s$.

\begin{figure}[h]
    \centering
    \includegraphics[width=0.75\linewidth]{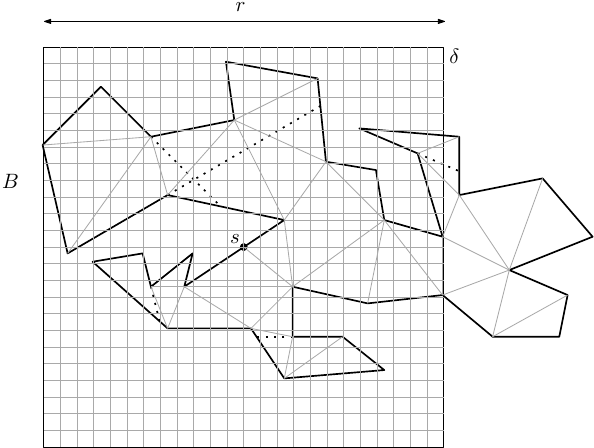}
    \caption{Discretization and localization of the general anchored $k$-\textsc{Watchman Routes} problem.}
\label{fig:discretization_general_simple_polygon}
\end{figure}

We argue that, for purposes of approximation, it suffices to restrict ourselves to grid-rounded solutions, that is, tours that have vertices among vertices of the decomposition. For an arbitrary $i =1 ,\ldots, k$, let $\sigma_1, \ldots, \sigma_m$ ($m\le 2n$) be the (closed) cells of the decomposition containing the vertices of $\gamma_i$. Some vertices of $\gamma_i$ may lie on the common boundary of multiple cells, in which case we can choose any of those cells; however, if a vertex of $\gamma_i$ is a reflection off of an essential cut, we choose a cell that is not in the pocket of the essential cut. The relative convex hull of $\gamma_i \cup \sigma_1 \cup \ldots\cup \sigma_m$ is the minimum-perimeter subset of $P$ that contains it, and must thus have perimeter no greater than $|\gamma_i| + |\partial\sigma_1| + \ldots + |\partial\sigma_m| \le |\gamma_i| + 8n\delta$. Moreover, since $\sigma_1, \ldots, \sigma_m$ are convex and enclose vertices of $\gamma_i$, the vertices of the relative convex hull $\gamma_i \cup \sigma_1 \cup \ldots\cup \sigma_m$ are vertices of $\sigma_1, \ldots, \sigma_m$.

Hence, for any $\varepsilon > 0$, if we let $\delta = \frac{\varepsilon L}{8nk}$, then there exists a grid-rounded tour $\gamma_i'$ such that $|\gamma_i'| \le (1 + \varepsilon)|\gamma_i|$. It is straightforward to argue that $\gamma'_i$ consists of consecutive geodesic shortest paths between grid points on some essential cuts (proof of Lemma~\ref{lem:discrete_viewpoints}).

\subsubsection{Fully polynomial-time approximation scheme} In fact, the very same dynamic programming algorithm for the anchored orthogonal $k$-\textsc{Watchman Routes} works to compute (approximately) $\{\gamma_i'\}_{i=1,\ldots, k}$, though we require some modifications, which we have actually introduced. In that algorithm, we have integers defining subproblems since Hanan grid points have integer coordinates, thus the length of any rectilinear path we consider is integral. In the general case with the $L_2$ metric, lengths and distances are potentially irrational even with the assumption that the vertices of $P$ have integer coordinates. Consequently, there is no natural increment we can set between consecutive values.

We compute the shortest single watchman route, let its length be $L$. Recall that we argued the following chain of inequalities for rectilinear routes, but the same can be said for the general case since orthogonality did not play a part in the argument:
\begin{align*}
    \frac{L}{k}\le \max\limits_{i=1,\ldots,k}|\gamma_i|\le L.
\end{align*}
We divide $L$ into $\lceil\frac{nk}{\varepsilon}\rceil$ uniform intervals and round the length of any geodesic shortest path the we consider in the algorithm down to the nearest interval endpoint. Using the same analysis of the approximation factor as in the rectilinear routes case, the dynamic program with subproblems defined by the intervals' endpoints returns a $(1 + \varepsilon)$-approximation to $\{\gamma'_i\}_{i=1,\ldots, k}$, which is a $(1 + \varepsilon)^2$-approximation to $\{\gamma_i\}_{i=1,\ldots, k}$. 

It remains to show that the dynamic programming algorithm with the above discretization of $P$ has fully polynomial running time (for fixed $k$). The number of grid lines within $P\cap B$ is $O\left(\frac{r}{\delta}\right) =  O\left(\frac{nk}{\varepsilon}\right)$. Therefore, the number of grid points on the essential cuts (intersections of essential cuts with triangulation chords and grid lines) is $O(n)\cdot O(n) + O(n)\cdot O\left(\frac{nk}{\varepsilon}\right) = O\left(\frac{n^2k}{\varepsilon}\right)$. There are in total $O\left(n\cdot\left(\frac{n^2k}{\varepsilon}\right)^k\cdot\left(\frac{nk}{\varepsilon}\right)^k\right) = O\left(\frac{n^{3k+1}}{\varepsilon^{2k}}\right)$ subproblems, each of which requires iterating through $O\left(\frac{n^2k}{\varepsilon}\right)$ previously solved subproblems. This results in an overall $O\left(\frac{n^{3k+3}}{\varepsilon^{2k+1}}\right)$ running time for the dynamic program, which computes a $(1 + \varepsilon)^2$-approximation. Note that if we let $2\varepsilon + \varepsilon^2 = \varepsilon'$, then $\frac{1}{\varepsilon} = \Theta\left(\frac{1}{\varepsilon'}\right)$ as $\varepsilon$ and $\varepsilon'$ approach 0, thus the running time is in the same order when written in terms of $\varepsilon'$.

We summarize the description of our algorithm as follows:
\begin{itemize}
    \item[Step 1:] Compute $L$, the length of the the shortest single watchman route for $P$.
    \item[Step 1:] Triangulate $P$.
    \item[Step 2:] Compute $r_{\min}$, the smallest value for which $GD_s(r_{\min})$ sees all of $P$. Set $r:=r_{\min}$.
    \item[Step 3:] Overlay a square grid of $\delta$-size pixels within a bounding box $B$ of size $r$-by-$r$ onto the triangulation of $P$ and the essential cuts, where $\delta = \frac{\varepsilon L}{8nk}$.
    \item[Step 4:] If $B$ intersects all essential cuts, execute the dynamic program, output the collection of routes $\{\gamma_i'\}_{i=1,\ldots, k}$ and store it.
    \item[Step 5:] If $r \le 6nr_{\min}$, then set $r:=2r$, then repeat from Step 3. Else, terminate the algorithm.  
\end{itemize}
We return the collection $\{\gamma_i'\}_{i=1,\ldots, k}$ that minimizes $\max\limits_{i = 1, \ldots, k}|\gamma_i'|$ out of all collections from all values of $r$ in the doubling search. Since the doubling search for $r$ takes $\log(6n) = O(\log n)$ iterations, our algorithm runs in time $O\left(\frac{n^{3k+3}}{\varepsilon^{2k+1}}\log n\right)$.

\begin{theorem}
    For any fixed $k$, the anchored $k$-\textsc{Watchman Routes} problem in a simple polygon has an FPTAS. 
\end{theorem}

\section{Multiwatchman Routes With a Variable Number of Watchmen}
\subsection{Seeing all of $P$}
The idea of localizing an optimal set of routes $\{\gamma_i\}_{i=1,\ldots, k}$ with a geodesic disk whose radius is no larger than $O\left(\max\limits_{i=1,\ldots,k}|\gamma_i|\right)$ gives us a simple constant-factor approximation. Similar to the previous section, suppose $r$ is such that $\frac{r}{2}\le\frac{\max\limits_{i=1,\ldots,k}|\gamma_i|}{2}\le r$, then the geodesic disk of radius $r$ centered on $s$, $GD_s(r)$ sees all of $P$. Let $\gamma \subset GD_s(r)$ be the shortest route that sees all of $P$ within $GD_s(r)$. Note that $\gamma$ is not necessarily the overall shortest watchman route of $P$, but we can still compute $\gamma$ using the polygon touring algorithm~\cite{dror2003touring, tan2018touring}. Any essential cut $c$ must intersect $GD_s(r)$, a relatively convex set, in a contiguous segment $c'$. By applying the polygon touring algorithm on the segments that are intersection of the essential cuts with $GD_s(r)$, we obtain $\gamma$. Divide $\gamma$ into $k$ subpaths of equal length, each of which is bounded by $a_i, a_{i+1}\in \gamma$ ($a_1 \equiv s \equiv a_{k+1}$) and denoted by $\gamma_{a_ia_{i+1}}$. Finally return the collection of routes $\{\gamma'_i\}_{i=1,\ldots, k}, \gamma'_i = \pi(s,a_i)\cup\gamma_{a_ia_{i+1}}\cup\pi(a_{i+1},s)$.

The approximation algorithm can be summarized in the following steps:
\begin{enumerate}
    \item[Step 1:] Compute $r_{\min}$, the smallest value for which $GD_s(r_{\min})$ sees all of $P$. Set $r:=r_{\min}$.
    \item[Step 2:] If $GD_s(r)$ intersects all essential cuts, compute $\gamma$, the shortest route within $GD_s(r)$ that sees all of $P$ using the polygon touring algorithm.
    \item[Step 3:] Divide $\gamma$ into $k$ subpaths of equal length, each bounded by $a_i, a_{i+1}\in \gamma$ ($a_1 \equiv s \equiv a_{k+1}$) and denoted by $\gamma_{a_ia_{i+1}}$.
    \item[Step 4:] For each $i$, we obtain $\gamma_i'$ by joining $\pi(s,a_i), \gamma_{a_ia_{i+1}}$ and $\pi(a_{i+1},s)$.
    \item[Step 5:] Set $r:=2r$, then repeat from Step 2, until $r > 6nr_{min}$.
\end{enumerate}
Finally, we return the collection of routes $\{\gamma_i'\}_{i=1,\ldots, k}$ that minimizes $\max\limits_{i = 1, \ldots, k}|\gamma_i'|$ out of all collections from all values of $r$ in the doubling search. 
\begin{theorem}
    The above algorithm has an approximation factor of 3.
\end{theorem}
\begin{proof}
    First, note that $\frac{|\gamma|}{k}\le \max\limits_{i = 1, \ldots, k}|\gamma_i| \le |\gamma|$.
    
    On the other hand, the geodesic distance from any point on $\gamma$ to $s$ is no longer than $r$. Thus, when $r \le \max\limits_{i=1,\ldots,k}|\gamma_i| \le 2r$, any $\gamma'_i$ returned by the algorithm is no longer than $\frac{|\gamma|}{k} + 2r \le  3\max\limits_{i=1,\ldots,k}|\gamma_i|$.
\end{proof}

\subsubsection{Analysis of running time} We determine $r_{\min}$ by computing the geodesic distance from $s$ to any of the $O(n)$ essential cuts in $O(n^2)$ time. For each value of $r$, we execute the $O(n^3)$ polygon touring algorithm~\cite{tan2018touring} and compute the $k$-collection $\{\gamma'_i\}_{i=1,\ldots, k}$ in $O(nk)$ time. The doubling search for $r$ takes $O(\log n)$ iterations (Lemma~\ref{lem:bounding_general_poly_see_everything}), thus the algorithm incurs a total running time of $O\left((n^3 + nk)\log n\right)$.

\subsubsection{Improving the approximation factor} In the approximation algorithm earlier, we gradually expand $GD_s(r)$ until $GD_s(r)$ contains an optimal $\{\gamma_i\}_{i=1,\ldots, k}$. If in each iteration, we instead multiply $r$ by a smaller positive number, such as $(1 + \varepsilon)$, then at some point $\frac{r}{(1 + \varepsilon)} \le \frac{\max\limits_{i=1,\ldots,k}|\gamma_i|}{2} \le r$. The distance from each point on $\gamma$ to $s$ is then no greater than $(1 + \varepsilon)\frac{\max\limits_{i=1,\ldots,k}|\gamma_i|}{2}$. Hence, the length of any of the $k$ routes returned by the approximation algorithm is bounded by $\frac{|\gamma|}{k} + 2(1 + \varepsilon)\frac{\max\limits_{i=1,\ldots,k}|\gamma_i|}{2} \le (2 + \varepsilon)\max\limits_{i=1,\ldots,k}|\gamma_i|$.

There is however, a trade-off between the approximation factor and the number of iterations of the multiplicative search for $r$. If we multiply $r$ by $(1 + \varepsilon)$ each time, the search requires $O(\log_{1+\varepsilon}n)$ iterations. Note that

\[\log_{1+\varepsilon}n = \log n \frac{\ln2}{\ln(1+\varepsilon)} = \log nO\left(\frac{1}{\varepsilon}\right).\]
In summary, we can achieve an approximation ratio of $(2 + \varepsilon)$ with a running time of $O\left(\frac{(n^3 + nk)\log n}{\varepsilon}\right)$.

\subsection{Seeing a quota of area within $P$}
Given an area quota $A$, let $\{\gamma_i\}_{i=1,\ldots, k}$ be such that $\left|\bigcup\limits_{i = 1,\ldots, k}V(\gamma_i)\right| \ge A$ and $\max\limits_{i=1,\ldots,k}|\gamma_i|$ is minimum. To discretize $P$ and localize $\{\gamma_i\}_{i=1,\ldots, k}$, we employ a technique similar as before, only more generalized. Instead of computing the smallest geodesic disk centered at $s$ that sees all of $P$, we compute the smallest disk that sees $A$, the given quota of area. Let $r_{\min}$ be the smallest value of $r$ such that $|V(GD_s(r)| = A$. We compute $r_{\min}$ using the ``visibility wave'' method in~\cite{quickest}. Consider the visibility graph of $P$, whose nodes are vertices of $P$ and two nodes are adjacent if they see one another. We sort the geodesic distance from $r$ to every visibility edge in increasing order in $O(n^2\log n)$ time. Then, propagate a ``wavefront'' through the sequences of visibility edges hit for the first time in the process of increasing $r$, stop when $|V(GD_s(r)| = A$. 

\begin{lemma}
    \label{lem:bounding_general_poly_quota}
    $r_{\min} \le \max\limits_{i=1,\ldots,k}|\gamma_i| = O(nr_{\min})$.
\end{lemma}
\begin{proof}
    Identical to that of Lemma~\ref{lem:bounding_general_poly_see_everything}.
\end{proof}

Lemma~\ref{lem:bounding_general_poly_quota} allows us to find $r$ such that $r\le\max\limits_{i=1,\ldots,k}|\gamma_i|\le 2r$ by a doubling search from $r = r_{\min}$ in $O(\log n)$ iterations. Suppose $\max\limits_{i=1,\ldots,k}|\gamma_i|\le 2r$, then $GD_s(r)$ contains $\{\gamma_i\}_{i=1,\ldots, k}$. We compute (approximately) a shortest tour within $GD_s(r)$ that sees at least an area of $A$, which we denote by $\gamma$. Such a route must exist, since $\{\gamma_i\}_{i=1,\ldots, k}$ collectively see an area no smaller than $A$.
\begin{lemma}
\label{lem:budget_watchman}
    \cite[Section 3]{huynh_et_al:LIPIcs.SWAT.2024.27}
    Given a budget $B \ge 0$ and any $\varepsilon > 0$, there exists an $O\left(\frac{n^5}{\varepsilon^6}\right)$ algorithm that computes a route of length at most $(1 + \varepsilon)B$ seeing as much area as any route of length $B$ within $GD_s(r)$.
\end{lemma}

We briefly describe the algorithm for completeness, and refer the readers to~\cite{huynh_et_al:LIPIcs.SWAT.2024.27} for more details. First, triangulate $P$, including $s$ as a vertex of the triangulation. Then, overlay onto the triangulation a regular square grid of side lengths $\delta = O\left(\frac{\varepsilon B}{n}\right)$ within an axis aligned square of size $B$-by-$B$ centered at $s$.

We consider the set of (convex) cells of the decomposition that overlap (both fully and partially) with $GD_s(r)$ and their vertices, $S_{\delta, r}$. Let $\gamma_B$ be the $B$-length route within $GD_s(r)$ that achieves the most area of visibility. Then, there exists a route of length at most $(1 + \varepsilon)B$ with vertices coming from $S_{\delta,r}$ enclosing $\gamma_B$, i.e. the boundary of the relative convex hull of the cells containing vertices of $\gamma_B$, thus seeing at least as much area as $\gamma_B$ (Figure~\ref{fig:approx_gamma_prime}).

\begin{figure}[h]
    \centering
    \includegraphics[width=0.8\textwidth]{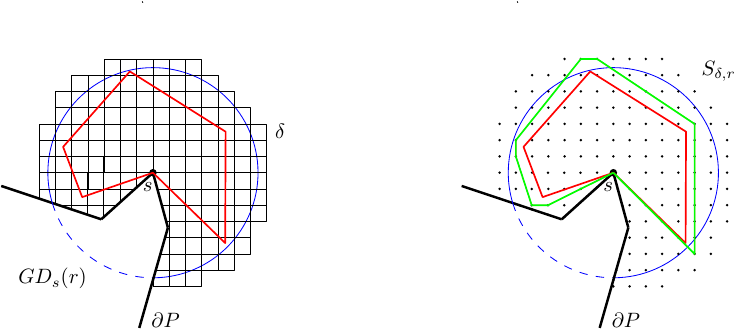}
    \caption{Left: $\gamma_B$ (red) is a tour no longer than $B$ within $C_g(r)$ (blue) that sees the most area. Right: enclosing $\gamma_B$ with a tour whose vertices are in $S_{\delta, r}$ seeing everything $\gamma_B$ sees (green).}
    \label{fig:approx_gamma_prime}
\end{figure}

The dynamic programming algorithm for the \textsc{Budgeted Watchman Route} problem can compute a relatively convex closed curve that sees the most with vertices among $S_{\delta,r}$. Sort $S_{\delta,r} = \{s_1, s_2, \ldots\}$ in (geodesic) angular order around $s$ to acquire a topological ordering of the subproblems. A subproblem $(s_j, \overline{B})$ is defined by a point in $S_{\delta,r}$ and a length. Let $a(s_j, \overline{B})$ be the maximum area of visibility that a relatively convex polygonal chain from $s$ to $s_j$ of length no greater than $\overline{B}$ can achieve, then
\[a(s_j, \overline{B}) = \max\limits_{i\le j} a(s_i, \overline{B} - |s_is_j|) + |V(s_is_j) \setminus V(\pi(s, s_i))|\]
(recall that $\pi(s,s_i)$ is the geodesic shortest path between $s$ and $s_i$).

Thus, if $|\gamma| \le B\le \alpha|\gamma|$ for some $\alpha \ge 1$, we can compute a route $\gamma'$ of length no longer than $\alpha(1 + \varepsilon)|\gamma|$ with vertices in $S_{\delta,r}$ that sees the most area, which must be no smaller than $|V(\gamma_B)| \ge |V(\gamma)| \ge A$. We show how to acquire a polynomial-sized set of values, from which we can ``guess'' $B$ so that $B$ is close to $|\gamma|$ (in particular, $|\gamma| \le B\le (1 + \varepsilon)|\gamma|$).

\begin{lemma}
\label{lem:bounding_quota_general}
    $r_{\min}\le |\gamma| \le 6nr_{\min}$.
\end{lemma}
\begin{proof}
    Identical to that of Lemma~\ref{lem:bounding_general_poly_see_everything}.
\end{proof}
We divide $6nr_{min}$ into $\lceil\frac{6n}{\varepsilon}\rceil$ uniform intervals so that each is no longer than $\varepsilon r_{min}$: the smallest interval endpoint that is no smaller than $|\gamma|$ must also be no larger than $(1 + \varepsilon)|\gamma|$, and hence is the value of $B$ that we desire. We perform a binary search on the values $\left\{0, \frac{6nr_{min}}{\lceil\frac{6n}{\varepsilon}\rceil}, \ldots, 6nr_{min} \right\}$ as the input budget for the algorithm in Lemma~\ref{lem:budget_watchman}, and seek out the smallest value for which the output route $\gamma'$ sees an area no smaller than $A$\@. Clearly, $|\gamma'|\le (1+\varepsilon)^2|\gamma|$.

We are now ready to describe the approximation algorithm for the \textsc{Quota $k$-Watchmen Routes} problem as follows:
\begin{itemize}
    \item[Step 1:] Set $r:= r_{\min}$.
    \item[Step 2:] Compute $\gamma'$, a $(1 + \varepsilon)^2$-approximation to $\gamma$.
    \item[Step 3:] Divide $\gamma'$ into $k$ subpaths of equal length, each bounded by $a_i, a_{i+1}\in \gamma'$ ($a_1 \equiv s \equiv a_{k+1}$) and denoted by $\gamma'_{a_ia_{i+1}}$.
    \item[Step 4:] For each $i$, we obtain $\gamma_i'$ by joining $\pi(s,a_i), \gamma'_{a_ia_{i+1}}$ and $\pi(a_{i+1},s)$.
    \item Step 5: Set $r:=2r$, then repeat from Step 2, until $r > 6nr_{min}$.
\end{itemize}
Return the collection of routes $\{\gamma_i'\}_{i=1,\ldots, k}$ that minimizes $\max\limits_{i = 1, \ldots, k}|\gamma_i'|$. 

    We argue that the algorithm described above has an approximation factor of $3 + \varepsilon$, where $\varepsilon$ is arbitrarily close to 0. First. we once again note that $\frac{|\gamma|}{k}\le \max\limits_{i = 1, \ldots, k}|\gamma_i'| \le |\gamma|$.
    
    Now, since all our choices for $B$ are no larger than $6nr_{\min}$, we can choose an appropriate $\delta = O\left(\frac{\varepsilon B}{n}\right)$ so that the geodesic distance from any point on $\gamma'$ to $s$ is no longer than $r + \varepsilon r$. Thus, when $\frac{r}{2} \le \frac{\max\limits_{i=1,\ldots,k}|\gamma_i|}{2} \le r$, any one of the $k$ routes returned by the algorithm is no longer than $\frac{|\gamma'|}{k} + 2r + 2\varepsilon r \le [(1 + \varepsilon)^2 + 2 + 2\varepsilon]\max\limits_{i=1,\ldots,k}|\gamma_i| = (3 + \varepsilon')\max\limits_{i=1,\ldots,k}|\gamma_i|$, where $\varepsilon' = 4\varepsilon + \varepsilon^2$. Since $\frac{1}{\varepsilon} = \Theta\left(\frac{1}{\varepsilon'}\right)$ as $\varepsilon$ and $\varepsilon'$ approach 0, the running time is in the same order when written in terms of either $\varepsilon$ or $\varepsilon'$.

For each choice of $B$, we execute the $O\left(\frac{n^5}{\varepsilon^6}\right)$ algorithm, thus computing an approximation to $\gamma'$ for each value of $r$ takes $O\left(\frac{n^5}{\varepsilon^6}\log\left(\frac{n}{\varepsilon}\right)\right)$ time. Deriving the collection $\{\gamma_i'\}_{i=1,\ldots, k}$ takes $O(nk)$ time. Since there are $O(\log n)$ iterations of the doubling search for $r$, the overall running time is $O\left(\left(\frac{n^5}{\varepsilon^6}\log\left(\frac{n}{\varepsilon}\right) + nk\right)\log n\right)$.

Similar to the $k$-\textsc{Watchman Routes} problem, we can improve the approximation, in particular, to $(2 + \varepsilon)$ by multiplying $r$ with $(1 + \varepsilon)$ instead of 2 in the multiplicative search. The running time in that case, is $O\left(\left(\frac{n^5}{\varepsilon^6}\log\left(\frac{n}{\varepsilon}\right) + nk\right)\frac{\log n}{\varepsilon}\right)$.

\begin{theorem}
    Both the anchored $k$-\textsc{Watchman Routes} problem and the anchored \textsc{Quota $k$-Watchman Routes} problem in a simple polygon have a polynomial-time $(2+\varepsilon)$ approximation.
\end{theorem}

\section{Conclusion and Future Work}
The \textsc{Watchman Route} problem is a classic problem in computational geometry, with many practical implications. Even though the single watchman version is well studied, the multiple watchmen version is much less so. We give tight approximability results when the number of watchmen is fixed and a common depot is specified, strengthening our theoretical understanding of the problem. The main source of intractability and approximation hardness of the problem seems to be from when we have multiple depots (or none at all), in which case the notion of essential cuts is not well-defined, and we do not have a good characterization of when a set of routes sees everything.

Even in the anchored setting, many questions remain open. Can we get a PTAS when $k$, the number of watchmen, is part of the input? Is there a fixed-parameter tractable algorithm? What about in a polygon with holes? We suspect new ideas and algorithmic techniques are needed to answer these questions.

\bibliographystyle{plainurl}
\bibliography{refs.bib} 

\end{document}